\documentclass[12pt]{article}

\usepackage[T1]{fontenc}
\usepackage[latin1]{inputenc}

\usepackage{amsmath,amsthm,amssymb,amscd,a4wide,url}

\numberwithin{equation}{section}

\usepackage{sectsty}
\sectionfont{\bf\large}
\subsectionfont{\bf\normalsize}

\newcommand{\ud}{\mathrm{d}}
\newcommand{\ui}{\mathrm{i}}
\newcommand{\ue}{\mathrm{e}}

\newcommand{\vb}{\boldsymbol{b}}

\newcommand{\vn}{\boldsymbol{n}}

\newcommand{\vx}{\boldsymbol{x}}
\newcommand{\vy}{\boldsymbol{y}}
\newcommand{\vz}{\boldsymbol{z}}
\newcommand{\vA}{\boldsymbol{A}}

\newcommand{\valpha}{\boldsymbol{\alpha}}

\newcommand{\gz}{{\mathbb Z}}
\newcommand{\rz}{{\mathbb R}}
\newcommand{\nz}{{\mathbb N}}

\newcommand{\R}{{\mathbb R}}

\DeclareMathOperator{\tr}{Tr}
\DeclareMathOperator{\mtr}{tr}

\numberwithin{equation}{section}

\newtheorem{theorem}{Theorem}[section]
\newtheorem{lemma}[theorem]{Lemma}

\newtheorem{cor}[theorem]{Corollary}

\theoremstyle{definition}
\newtheorem{defn}[theorem]{Definition}

\begin{document}

\thispagestyle{empty}

\hfill \phantom{\tt version \today}

\vspace*{4ex}

\noindent
{\large\bf Heat kernel asymptotics for magnetic Schrödinger operators}\\[3ex]
{\bf Jens Bolte${}^a$ and Stefan Keppeler}${}^b$\\[3ex]
${}^a${\small{Department of Mathematics, Royal Holloway, University of London\\
\phantom{${}^a$}Egham, TW20 0EX, United Kingdom\\
\phantom{${}^a$}{\tt jens.bolte@rhul.ac.uk}}}\\[3ex]
${}^b${\small{Mathematisches Institut, Universität Tübingen,\\
\phantom{${}^b$}Auf der Morgenstelle 10, 72076 Tübingen, Germany\\
\phantom{${}^b$}{\tt stefan.keppeler@uni-tuebingen.de}}}\\[3ex]
\parbox{.85\textwidth}{
{\bf Abstract.} We explicitly construct parametrices for magnetic
Schrödinger operators on $\R^d$ and prove that they provide a complete
small-$t$ expansion for the corresponding heat kernel, both on and off
the diagonal.}

\newpage

\section{Introduction}
\label{sec:Einl}
Heat kernel asymptotics have attracted much attention ever since
Minakshisundaram and Pleijel \cite{MinPle49,Min53} proved that the 
heat kernel for the Laplacian on a compact Riemannian manifold has 
a complete asymptotic expansion as $t\to 0^+$. These asymptotic 
expansions have since been extended to many elliptic operators of 
geometric relevance (see, e.g., 
\cite{BerGauEdm71,BerGetVer92,Gil95,Kir01}). Most notably, heat
kernel expansions reveal the local nature of heat invariants, i.e.,
the coefficients in the small-$t$ asymptotics of the trace of the heat
semi-group generated by the operator in question. Among other
things, this observation led to novel proofs of index theorems for 
elliptic complexes \cite{AtiBotPat73}. In general, heat kernel 
expansions turned out to provide powerful tools in spectral geometry. 

One would expect similar properties of heat kernels for other
(semi-bounded, self-adjoint) operators, too. Schr\"odinger operators
are among the most prominent examples. In particular, magnetic
Schr\"odinger operators are of interest in the context of para- and
diamagnetism where heat kernel estimates, including their small-$t$
behaviour, were used very successfully (see, e.g.,
\cite{LosTha97,Erd97}).

In a scattering context, where a Schr\"odinger operator has a
non-empty absolutely continuous spectrum and, therefore, the heat
semi-group is not of trace class, a `relative heat trace', as it
appears in a Krein trace formula, has been shown to posses a related
small-$t$ expansion (see \cite{Col81} for the non-magnetic case and
\cite{Hit02} for magnetic Schr\"odinger operators). Moreover,
irrespective of whether the heat semi-group of a Schr\"odinger
operator is of trace class, small-$t$ asymptotics of the diagonal of
the heat kernel have been proven (in \cite{HitPol03} for non-magnetic
operators and in \cite{KorPus03} for the magnetic case). In that
context the heat invariants are integrals of the scalar potential and 
the magnetic field strength as well as their derivatives.

In this paper we consider magnetic Schr\"odinger operators on $\rz^d$
with smooth and polynomially bounded scalar and vector
potentials. Our principal goal is to prove complete asymptotic
expansions as $t\to 0^+$ of the related full heat kernels (and not
only of their diagonals as, e.g., in \cite{KorPus03}). In order to
achieve this we construct parametrices for the heat equations and show
that these provide heat kernel asympotics. We then use the
parametrices to determine heat invariants for magnetic Schr\"odinger
operators on some manifolds with and without boundary of the form
$\rz^2/\Gamma$, where $\Gamma$ is a discrete group of reflections and 
translations. The paper is organised as follows. We
first briefly describe the setting of heat kernels and parametrices
for magnetic Schr\"odinger operators in Section \ref{sec:prelim}. The
parametrices are then constructed in Section \ref{sec:Para} and the
main result is summarised in Theorem \ref{prop:Para}. In Section
\ref{sec:Volt} we use the parametrices to construct a Volterra series
and prove that this series yields the heat kernel. Morever, this
series turns out to provide an asymptotic expansion for small
$t$. These results are contained in Theorem
\ref{thm:Volterra}. Finally, in Section \ref{sec:appl} we discuss some
applications to heat trace invariants for magnetic Schr\"odinger
operators on half-planes, infinite cylinders, and tori.

\section{Preliminaries}
\label{sec:prelim}
We consider magnetic Schr\"odinger operators
\begin{equation}
\label{magSchroe}
 H = \bigl( -\ui\nabla - \vA(\vx) \bigr)^2 + V(\vx) \ ,
\end{equation}
acting on a suitable domain in $L^2(\rz^d)$ with $d\geq 2$. For our purposes
we require the components $A_j$ of the vector potential $\vA$ and the 
potential $V$ to be in 
\begin{equation}
\label{polybound}
 \mathcal{S}_m := 
 \bigl\{ f \in C^\infty(\rz^d) ; \ \forall \alpha \in \nz_0^d \ 
         \exists\, C_\alpha > 0 \ s.t.\ 
         |\partial_x^\alpha f(\vx)| \leq C_\alpha (1+|\vx|)^m\bigr\}
\end{equation}
for some fixed $m\geq 0$. Moreover, $V$ shall be bounded from below. Under
these conditions $H$ is essentially self-adjoint on the domain 
$C_0^\infty(\rz^d)$, which is indeed true under much weaker conditions
(see, e.g., \cite{Sim73,LeiSim81}).
The vector potential defines a one-form
\begin{equation}
 A(\vx) = \sum_{k=1}^d A_k (\vx)\,\ud x_k\ ,
\end{equation}
whose exterior derivative
\begin{equation}
 B(\vx) := \ud A(\vx) = \sum_{k<l} B_{kl}(\vx)\,\ud x_l\wedge\ud x_k\ ,
\end{equation}
is the two-form of magnetic field strengths, with components
\begin{equation}
\label{eq:Bcompdef}
 B_{kl}(\vx) = \frac{\partial A_k}{\partial x_l}(\vx) -
 \frac{\partial A_l}{\partial x_k}(\vx)
\end{equation}
that are also in $\mathcal{S}_m$.

The heat kernel for $H$ is a function 
$K\in C^\infty\bigl((0,\infty)\times\rz^d\times\rz^d\bigr)$ satisfying
\begin{equation}
\label{eq:heat_kernel_def}
 \left( \frac{\partial}{\partial t} + H \right) K(t,\vx,\vy) = 0 
 \ , \qquad \lim_{t\to0+} K(t,\vx,\vy) = \delta (\vx-\vy)\ .
\end{equation}
The existence of the heat kernel is known and follows, e.g., from a
representation in terms of a Feynman-Kac formula (see, e.g., \cite{Sim05}). 

Often one is interested in the small-$t$ behaviour of the heat kernel or of 
its trace, if the heat semi-group is of trace class. In that case one can infer
the distribution of eigenvalues from the asymptotics of the heat semi-group.
It is well known that for large classes of operators $H$ the heat trace allows 
for an asymptotic expansion of the form
\begin{equation}
 \tr\ue^{-Ht} = \frac{1}{(4\pi t)^{d/2}}\sum_{j=0}^N\alpha_j\,t^j + 
 O(t^{-\frac{d}{2}+N+1})\ .
\end{equation}
Operators for which this is known to hold include Laplacians on bounded
domains in $\rz^d$, on compact Riemannian manifolds with and without 
boundary (see, e.g., \cite{MinPle49,Kac66,McKSin67,Gil95}), and certain 
magnetic and non-magnetic Schr\"odinger operators  
(see, e.g., \cite{HitPol03,KorPus03}).

The heat trace coefficients can be obtained from the local heat invariants,
\begin{equation}
 \alpha_j = \int a_j(\vx)\ \ud x\ ,
\end{equation}
which are determined by the small-$t$ asymptotics of the diagonal,
\begin{equation}
\label{eq:heatdiagexp}
 K(t,\vx,\vx) = \frac{1}{(4\pi t)^{d/2}}\sum_{j=0}^N a_j(\vx)\,t^j + 
 O(t^{-\frac{d}{2}+N+1})\ .
\end{equation}

It is our intention to determine a complete small-$t$ asymptotics for
the heat kernel of a magnetic Schr\"odinger operator
\eqref{magSchroe}, and to compute the first few coefficients
explicitly. To achive this we shall use a parametrix for the heat
equation, i.e., an approximate solution of \eqref{eq:heat_kernel_def}
for small $t$. A more precise definition is as follows.
\begin{defn}
\label{def:para}
A parametrix of order $\lambda\geq 0$ for the heat equation is a function 
$k\in C^\infty\bigl((0,\infty)\times\rz^d\times\rz^d\bigr)$ such that 
\begin{enumerate}
\item[(i)] $\bigl(\frac{\partial}{\partial t} + H\bigr) 
 k(t,\vx,\vy)$ extends to a function in 
 $C^\infty\bigl([0,\infty)\times\rz^d\times\rz^d\bigr)$,
\item[(ii)] there exists a constant $c_\lambda >0$ such that 
 $\bigl|\bigl(\frac{\partial}{\partial t} + H\bigr) k(t,\vx,\vy)\bigr|
 \leq c_\lambda\,t^\lambda$,
\item[(iii)] $\lim\limits_{t\to0+} k(t,\vx,\vy) = \delta (\vx-\vy)$.
\end{enumerate}
\end{defn}
Parametrices allow to read off the expansion \eqref{eq:heatdiagexp}, including
the heat trace invariants $a_j(\vx)$. As they approximate the heat kernel
itself, and not only its diagonal, parametrices can be used to determine
heat trace asymptotics in cases where the heat kernel is determined from a 
given kernel using a method of images.

\section{Parametrix construction}
\label{sec:Para}
In order to obtain a heat parametrix we introduce the ansatz
\begin{equation}
\label{ansatz}
 k_N(t,\vx,\vy) = \frac{1}{(4\pi t)^{d/2}}\,\ue^{-\frac{1}{t}\phi(\vx,\vy)} 
 \sum_{k=0}^{N+1} u_k(\vx,\vy)\,t^k \ ,\qquad\text{where}\ u_0\not\equiv 0\ ,
\end{equation}
and show that one can solve recursively for 
$\phi,u_0,u_1,\ldots\in C^\infty (\rz^d\times\rz^d)$. This procedure will 
eventually provide us with a complete asymptotic small-$t$ expansion of 
the heat kernel.
\begin{theorem}
\label{prop:Para}
Let $N\in\nz_0$ be such that $N\geq d/2 -1$ and define
\begin{equation}
\label{paraNdef}
 k_N(t,\vx,\vy) = \frac{1}{(4\pi t)^{d/2}}\,\ue^{-\frac{1}{4t}(\vx-\vy)^2} 
 \sum_{k=0}^{N+1} u_k(\vx,\vy)\,t^k \ ,
\end{equation}
with leading coefficient
\begin{equation}
\label{eq:0solution}
  u_0(\vx,\vy) = \exp\left( 
    \ui \int_0^1 \vA(\vx(s))\cdot\dot\vx(s) \, \ud s \right) \,
\end{equation}
and (recursively defined) higher-order coefficients
\begin{equation}
\label{u_krecur}
  u_k(\vx,\vy) = -u_0(\vx,\vy)\int_0^1 s^{k-1}\,u_0^{-1}(\vx(s),\vy)\,(Hu_{k-1})
  (\vx(s),\vy) \  \ud s \ ,\qquad k\geq1\ ,
\end{equation}
where $\vx(s)=\vy+s(\vx-\vy)$. Then \eqref{paraNdef} is a parametrix 
of order $N+1-d/2$ for the heat equation.
\end{theorem}
\begin{proof}
Inserting the ansatz \eqref{ansatz} into \eqref{eq:heat_kernel_def} and 
ordering terms by powers of $t$ we obtain
\begin{equation}
\begin{split}
  \left( \frac{\partial}{\partial t} + H \right)  k_N
  = \frac{\ue^{-\frac{1}{t}\phi}}{(4\pi t)^{d/2}} \, \sum_{k=0}^{N+1} \Bigl\{ & 
     t^{k-2} \, \bigl[ \phi - (\nabla\phi)^2 \bigr] \, u_k 
     + t^{k-1} \, \Bigl[ \Delta \phi 
        + \Bigl(k-\frac{d}{2}\Bigr) \\ 
    & + 2\,(\nabla\phi) \cdot (\nabla-\ui\vA) \Bigr] \, u_k 
    + t^k \, H \, u_k \Bigr\} \, .
\end{split}
\end{equation}
Requiring the coefficients of $t^{j-d/2}$, $j=-2,\hdots,N$, to vanish
independently yields the following hierarchy of conditions,
\begin{align}
 t^{-2-d/2}&: && \phi - (\nabla\phi)^2 = 0 \, ,
 \\ \label{eq:pre-transport_u0}
 t^{-1-d/2}&: && \bigl[ \phi - ( \nabla\phi)^2 \bigr] u_1 
 + \Bigl[ \Delta\phi - \frac{d}{2} + 2\,(\nabla\phi)\cdot (\nabla-\ui\vA) 
 \Bigr] u_0 = 0 \, ,
 \\ \label{eq:pre-transport_uk}
 0\leq k\leq N&: && \bigl[ \phi - ( \nabla\phi)^2 \bigr] u_{k+2} 
 + \Bigl[ \Delta\phi + k+1-\frac{d}{2} + 2\,(\nabla\phi)\cdot (\nabla-\ui\vA) 
 \Bigr] u_{k+1} + H u_k = 0 \, .
\end{align}
The solution to the first equation is 
\begin{equation}
  \phi(\vx,\vy) = \frac{1}{4} (\vx+\vb)^2
\end{equation}
with $\vb \in \R^d$. The initial condition for $t\to 0+$ for the parametrix 
implies $\vb=-\vy$, as well as $u_0(\vx,\vx)=1$. This means that 
$\nabla\phi(\vx,\vy)=\frac{1}{2}(\vx-\vy)$ and $\Delta\phi=\frac{d}{2}$.

Substituting the solution for $\phi$ into 
eqs.~\eqref{eq:pre-transport_u0} 
and \eqref{eq:pre-transport_uk} we
find a homogenous transport equation
\begin{equation}
 \label{eq:transport_u0}
 (\vx-\vy) \cdot (\nabla-\ui\vA) \, u_0(\vx,\vy) = 0
\end{equation}
for the lowest order coefficient, and inhomogenous transport equations 
\begin{equation}
\label{eq:transport_uk}
 \bigl[ (\vx-\vy) \cdot (\nabla-\ui\vA) + k \bigr] \, u_k(\vx,\vy) 
 = - H u_{k-1}(\vx,\vy)
 \quad \forall\ k\geq1 \ ,
\end{equation}
for the higher-order coefficients. 

In order to solve these equations we introduce the parametrisation
\begin{equation}
\label{curve}
 \vx(s) := \vy + s(\vx-\vy) \ ,\quad 0\leq s\leq 1 \ ,
\end{equation}
of the line connecting $\vy$ and $\vx$. Hence, $\vx(0)=\vy$,
$\vx(1)=\vx$ and $\dot{\vx}(s)=\vx-\vy$. Thus,
\begin{equation}
  \left. \left(\frac{\ud}{\ud s} - \ui \dot{\vx}(s) \cdot \vA(\vx(s)) \right)
  u_k(\vx(s),\vy) \right|_{s=1} = 
  (\vx-\vy) \cdot \bigl( \nabla - \ui \vA (\vx) \bigr) u_k(\vx,\vy) \, .
\end{equation}
The homogeneous and inhomogeneous ordinary differential equations for 
$u_k(\vx(s),\vy)$ can be solved explicitly, providing solutions to the 
transport equations in the form $u_k(\vx,\vy)=u_k(\vx(1),\vy)$. (See 
\cite{Yos53} for a related approach.)

For the lowest order (homogeneous) transport equation
\eqref{eq:transport_u0} we have to solve 
\begin{equation}
  \left( \frac{\ud}{\ud s} - \ui \dot{\vx}(s) \cdot \vA(\vx(s)) \right)
  u_0(\vx(s),\vy) = 0 \, , \qquad 
  u_0(\vx(0),\vy) = 1 \, .
\end{equation}
The solution can readily be found to be
\begin{equation}
  u_0(\vx(s),\vy) 
  = \exp\left( \ui \int_0^s \vA(\vx(t)) \cdot \dot{\vx}(t) \, \ud t \right) \, .
\end{equation}
Choosing $s=1$, this gives the lowest-order coefficient $u_0$ in the 
form \eqref{eq:0solution}.

The inhomogeneous, higher-order transport equations
\eqref{eq:transport_uk} can be solved recursively, order by order. We
first determine the general solutions
\begin{equation}
\label{eq:homansatz}
  u_k^\mathrm{hom}(\vx,\vy) = u_0(\vx,\vy) \, v_k(\vx,\vy) 
\end{equation}
of the homogeneous equations corresponding to 
\eqref{eq:transport_uk}. The functions $v_k$ hence follow from 
the condition
\begin{equation}
\label{eq:homkeq}
 \bigl[ (\vx-\vy)\cdot\nabla +k \bigr] v_k(\vx,\vy) = 0 \ ,
\end{equation}
implying
\begin{equation}
\label{eq:0homsol}
 v_k(\vx,\vy) = \frac{f_k(\vy)}{|\vx-\vy|^k} \, ,
\end{equation}
with arbitrary $f_k$. As we require $u_k\in C^\infty(\rz^d\times\rz^d)$,
we conclude that $u_k^\mathrm{hom}\equiv 0$.

In order to determine solutions of the inhomogeneous  transport 
equation \eqref{eq:transport_uk} we introduce the ansatz
\begin{equation}
\label{1solution}
 u_k(\vx,\vy) = u_0(\vx,\vy) \, w_k(\vx,\vy) \ ,
\end{equation}
where it is understood that $w_0\equiv1$. The inhomogeneity requires us to 
act with $H$ on $u_{k-1}$, given in the form \eqref{1solution}. As a first 
step, a straightforward calculation yields that
\begin{equation}
\label{eq:Bverifymod}
\begin{split}
  \left( \frac{\partial}{\partial x_j} - \ui A_j (\vx) \right) u_0(\vx,\vy) 
  = -\ui\sum_{l=1}^d (x_l-y_l)\,\int_0^1 t B_{jl}(\vx(t)) \, \ud t \
    u_0(\vx,\vy) \, .
\end{split}
\end{equation}
Using this, the inhomogeneity can be brought into the form
\begin{equation}
\begin{split}
\label{inhomg}
  -H u_0(\vx,\vy) \, w_{k-1}(\vx,\vy) 
  &= -\Bigl(\bigl( -\ui\nabla - \vA(\vx) \bigr)^2 + V(\vx) \Bigr) 
     \, u_0(\vx,\vy) \, w_{k-1}(\vx,\vy) \\
  &= u_0(\vx,\vy)\, g_{k-1}(\vx,\vy)\, ,
\end{split}
\end{equation}
where, following a straightforward calculation,
\begin{equation}
\begin{split}
\label{g_kdef}
 g_k (\vx,\vy) = 
   &-\Bigg[ -\Delta + V(\vx) + \ui (\vx-\vy)\cdot\valpha(\vx,\vy) + \ui 
     (\vx-\vy)\cdot\beta(\vx,\vy)\nabla \\
   & \quad + (\vx-\vy)\cdot\gamma(\vx,\vy)(\vx-\vy)\Bigg] w_k(\vx,\vy) \, ,
\end{split}
\end{equation}
and we defined the vector $\valpha(\vx,\vy)$ with components
\begin{equation}
\label{def_alpha}
 \alpha_l(\vx,\vy) 
 = \sum_{j=1}^d \int_0^1 t^2 
   \, \frac{\partial B_{jl}}{\partial x_j}(\vx(t)) \, \ud t\, ,
\end{equation}
and the matrices $\beta(\vx,\vy)$ and $\gamma(\vx,\vy)$ with entries
\begin{equation}
\label{def_beta_gamma}
\begin{split}
 \beta_{jl}(\vx,\vy) &:= 2\int_0^1 t B_{jl}(\vx(t)) \, \ud t\, , \\
 \gamma_{jl}(\vx,\vy) &:= \sum_{m=1}^d \int_0^1\int_0^1 ts B_{ml}(\vx(t))
                          B_{mj}(\vx(s))\ \ud t\,\ud s\ .
\end{split}
\end{equation}
The functions $\alpha_l,\beta_{jl},\gamma_{jl}$ are smooth and
polynomially bounded.

Furthermore, using \eqref{eq:transport_u0} on the left-hand side of
\eqref{eq:transport_uk} yields 
\begin{equation}
 \Bigl( (\vx-\vy)\cdot\bigl( \nabla - \ui \vA (\vx)\bigr) + k \Bigr) 
 u_0(\vx,\vy) \, w_k(\vx,\vy) =  
 u_0(\vx,\vy)\bigl( (\vx-\vy)\cdot\nabla + k \bigr) w_k(\vx,\vy) \ .
\end{equation}
Hence, the functions $w_k$ follow from
\begin{equation}
 \bigl( (\vx-\vy)\cdot\nabla + k \bigr) w_k(\vx,\vy) = g_{k-1}(\vx,\vy) \, .
\end{equation}
With the ansatz
\begin{equation}
\label{khomsol}
 w_k(\vx,\vy) = \frac{f_k(\vx,\vy)}{|\vx-\vy|^k} \, ,
\end{equation}
the function $f_k$, therefore, satisfies the equation
\begin{equation}
 \frac{1}{|\vx-\vy|^k} (\vx-\vy)\cdot\nabla f_k(\vx,\vy) = g_{k-1}(\vx,\vy) \, .
\end{equation}
Along the curve \eqref{curve} this reads
\begin{equation}
 \frac{s^{1-k}}{|\vx-\vy|^k} (\vx-\vy)\cdot(\nabla f_k)(\vx(s),\vy) = 
 g_{k-1}(\vx(s),\vy) \, ,
\end{equation}
or
\begin{equation}
 \frac{\ud}{\ud s}\,f_k(\vx(s),\vy) = 
 s^{k-1} |\vx-\vy|^k \, g_{k-1}(\vx(s),\vy)\, ,
\end{equation}
and is solved by
\begin{equation}
 f_k(\vx,\vy) 
 = f_k(\vy,\vy) + |\vx-\vy|^k \int_0^1 s^{k-1} \, g_{k-1}(\vx(s),\vy) \, \ud s \ ,
\end{equation}
i.e.
\begin{equation}
 w_k(\vx,\vy) = \frac{f_k(\vy,\vy)}{|\vx-\vy|^k} +
 \int_0^1 s^{k-1} \, g_{k-1}(\vx(s),\vy) \, \ud s \, .
\end{equation}
In order to avoid a singularity on the diagonal we have to choose
$f_k(\vy,\vy)=0$, hence
\begin{equation}
\label{w_kresult}
 w_k(\vx,\vy) 
 = \int_0^1 s^{k-1} g_{k-1}(\vx(s),\vy) \, \ud s \ 
\end{equation}
and 
\begin{equation}
\label{u_kresult}
 u_k(\vx,\vy) 
 = u_0(\vx,\vy) \,\int_0^1 s^{k-1} g_{k-1}(\vx(s),\vy) \, \ud s \ .
\end{equation}
With \eqref{inhomg}, this implies \eqref{u_krecur}.

Finally, if $N\geq\frac{d}{2}-1$,
\begin{equation}
\begin{split}
\label{R_Ndef}
 R_N (t,\vx,\vy)
  &:=\Bigl(\frac{\partial}{\partial t} + H\Bigr) k_N(t,\vx,\vy) 
          = \frac{t^{N+1}}{(4\pi t)^{d/2}}\,\ue^{-\frac{1}{4t}(\vx-\vy)^2} 
          Hu_{N+1}(\vx,\vy)  \\
  &= -\frac{t^{N+1-d/2}}{(4\pi)^{d/2}}\,\ue^{-\frac{1}{4t}(\vx-\vy)^2} \,
       u_0(\vx,\vy)\,g_{N+1}(\vx,\vy) 
 \end{split}
\end{equation}
extends to $t=0$ such that property (i) in Definition~\ref{def:para}
is satisfied; property (ii) can be read off too and property (iii) was
built in as an initial condition for the lowest-order transport equation.
\end{proof}
We remark that the higher-order coefficients $u_k$ have to be determined 
recursively; we only do this explicitly for $u_1$. As $w_0=1$, eqs.\ 
\eqref{u_kresult} and \eqref{g_kdef} with $k=1$ give
\begin{equation}
\begin{split}
\label{u_1result}
 u_1(\vx,\vy)  = \int_0^1 \Bigg[ - 
   &V(\vx(s))-\ui s\sum_{l=1}^d (x_l-y_l)\alpha_l(\vx(s),\vy) \\
   &-s^2\sum_{j,l=1}^d (x_l-y_l) \gamma_{lj}(\vx(s),\vy)(x_j-y_j)
   \Bigg] \ \ud s \  u_0(\vx,\vy) \ .
\end{split}
\end{equation}
The diagonal terms simplify considerably,
\begin{equation}
 u_k(\vx,\vx) = g_{k-1}(\vx,\vx) \, \int_0^1s^{k-1}\ \ud s 
 = \left. \frac{1}{k}\bigl(\Delta-V(\vx)\bigr)
   w_{k-1}(\vx,\vy)\right|_{\vy=\vx}\ .
\end{equation}
They yield the well-known local heat invariants (compare, e.g., 
\cite{KorPus03}); the lowest orders are
\begin{equation}
\label{eq:loworders}
 u_1(\vx,\vx)  = -V(\vx)\ ,\quad  u_2(\vx,\vx) = \frac{1}{2} V^2(\vx) 
   - \frac{1}{6} \Delta V(\vx)
   + \frac{1}{12} \mtr B^2(\vx) \ .
\end{equation}

The full heat kernel for a magnetic Schr\"odinger operator is determined
by the potentials $\vA$ and $V$ via \eqref{eq:heat_kernel_def}. A parametrix,
however, follows from the 'local' transport equations and hence depends
only on the potentials along the straight line connecting $\vx$ and $\vy$.
\begin{cor}
\label{cor:locality}
Let $k_N$ and $\tilde k_N$ be parametrices (of the same order $N$)
corresponding to potentials $V,\vA$ and $\tilde V,\tilde\vA$ in the heat
equation, respectively. Let $\vx,\vy$ be given and assume that 
$V=\tilde V$ and $\vA=\tilde\vA$ in a neighbourhood of the straight line 
$\{\vx(s);\ 0\leq s\leq 1\}$ connecting $\vy$ and $\vx$. Then 
$k_N(t,\vx,\vy)=\tilde k_N(t,\vx,\vy)$.
\end{cor}

To put our result into perspective we compare it with the well-known
case of a constant magnetic field $B$ in dimension $d=2$. There the
heat kernel for the magnetic Schr\"odinger operator is given by the
Mehler kernel (see, e.g., \cite{Sim05,LosTha97}),
% (Their time variable has to be rescaled 
% according to $t \mapsto 2t$ in order 
% to match our conventions.)
\begin{equation}
\label{Mehler_kernel}
  K(t,\vx,\vy) 
  = \frac{B}{4\pi\sinh(Bt)} \, \exp \left\{
    -\frac{B}{4} \coth(Bt) (\vx-\vy)^2 
         - \ui \frac{B}{2} (x_1 y_2 - x_2 y_1) \right\} \, ,   
\end{equation}
where the gauge 
\begin{equation}
\label{symmetric_gauge}
  \vA(\vx) = \frac{B}{2} \begin{pmatrix} -x_2 \\ x_2 \end{pmatrix} 
\end{equation}
has been chosen. Expanding the Mehler kernel for small $t$ yields
\begin{equation}
  K(t,\vx,\vy) 
  = \frac{1}{4\pi t} \, \ue^{-\frac{1}{4t}(\vx-\vy)^2} 
    \ue^{-\ui \frac{B}{2} (x_1 y_2 - x_2 y_1)} 
    \left( 1 - \frac{B^2}{12} (\vx-\vy)^2 \, t + O(t^2) \right) \, .
\end{equation}
Inserting the vector potential \eqref{symmetric_gauge} into
\eqref{eq:0solution} we find
\begin{equation}
  u_0(\vx,\vy) =  \ue^{-\ui \frac{B}{2} (x_1 y_2 - x_2 y_1)} \, .
\end{equation}
The auxilliary functions \eqref{def_alpha} and \eqref{def_beta_gamma}
become $\alpha = \gamma_{12} = \gamma_{21} = 0$ and $\gamma_{11} =
\gamma_{22} = B^2/4$. Inserting into \eqref{u_1result} yields
\begin{equation}
  u_1(\vx,\vy) = -\frac{B^2}{12}(\vx-\vy)^2 \, u_0(\vx,\vy) \, , 
\end{equation}
i.e.\ the parametrix $k_0$ agrees with the leading small-$t$
asymptotics of the Mehler kernel. In the following section we show that it
is always true that the parametrix $k_N$ provides the first $N+1$
terms of the heat kernel asymptotics.

\section{Asymptotic expansion of the heat kernel}
\label{sec:Volt}
From the ansatz \eqref{ansatz} one expects a heat parametrix to be an 
approximation to the heat kernel for small $t$. It even appears to be a tool
to generate an asymptotic expansion for the heat kernel. That this is indeed
the case can be proven by generating a Volterra series from a sufficiently
regular parametrix, and then to prove that this series converges. This then
first implies the existence of the heat kernel, since the Volterra series turns 
out to represent the heat kernel. As the existence of the heat kernel is known 
by other means, here the more important consequence of the Volterra series 
is that it provides an asymptotic expansion of the heat kernel for $t\to 0$.

The approach we follow to construct the Volterra series exists in many 
variants; it was developed for heat kernels of elliptic operators on manifolds 
\cite{Min53,Yos53}, and was then further extended in many directions. We 
here take some inspiration from \cite{Gri04}.

Recall that $k_N$ as given in Theorem~\ref{prop:Para} is a parametrix of order 
$\lambda=N+1-d/2$ for the heat equation of a magnetic Schr\"odinger operator, 
compare \eqref{R_Ndef}. This requires $N\geq d/2 -1$, i.e., in two dimensions
the simplest choice is $N=0$, leading to
\begin{equation}
 k_0 (t,\vx,\vy) = \frac{1}{4\pi t}\,\ue^{-\frac{1}{4t}(\vx-\vy)^2}\,
 \bigl(u_0(\vx,\vy)+u_1(\vx,\vy)\,t \bigr)
\end{equation}
and, cf.~\eqref{R_Ndef},
\begin{equation}
 R_0 (t,\vx,\vy) = -\frac{1}{4\pi}\,\ue^{-\frac{1}{4t}(\vx-\vy)^2}\,
 u_0(\vx,\vy)g_1(\vx,\vy)\ .
\end{equation}
In higher dimensions further terms are required.

As a first step towards constructing a Volterra series we need to estimate
$R_N$.
\begin{lemma}
\label{lem:R_Nestimate}
Let $V,A_j\in\mathcal{S}_m$ for some $m\geq 0$, cf.\
\eqref{polybound}, and let $N\in\nz_0$. Then, for every $\varepsilon>0$
there exists $M_N>0$ such that
\begin{equation}
\label{R_0est}
 |R_N (t,\vx,\vy)| 
 \leq M_N\,t^{N+1-d/2}\,\ue^{-\frac{1-\varepsilon}{4t}(\vx-\vy)^2}\ .
\end{equation}
\end{lemma}
\begin{proof}
According to \eqref{R_Ndef}, we have to estimate $g_{N+1}(\vx,\vy)$. This can
be done using \eqref{g_kdef}, and taking into account that 
\begin{equation}
\label{wkrecurs}
 w_k (\vx,\vy) = \frac{u_k(\vx,\vy)}{u_0(\vx,\vy)}
 = \int_0^1 s^{k-1} \, g_{k-1}(\vx(s),\vy) \, \ud s\, .
\end{equation}
From \eqref{g_kdef}, and the fact that the potentials are assumed to
be smooth and polynomially bounded, we see that
\begin{equation}
 g_0(\vx,\vy) = -V(\vx) - \ui(\vx-\vy)\cdot\valpha(\vx,\vy) -
  (\vx-\vy)\cdot\gamma(\vx,\vy)(\vx-\vy)
\end{equation}
is also smooth and polynomially bounded. For $k\geq 1$ the functions $g_k$ 
are determined recursively through \eqref{g_kdef}, \eqref{w_kresult} and 
\eqref{wkrecurs}. That way they are also seen to be smooth and polynomially 
bounded. Hence the estimate \eqref{R_0est} follows.
\end{proof} 
To proceed we need the convolution of kernels.
\begin{defn}
Let $f,g\in C\bigl((0,\infty)\times\rz^d\times\rz^d\bigr)$. Then their
convolution is
\begin{equation}
 (f * g) (t,\vx,\vy) := \int_0^t\int_{\rz^d}f(t-s,\vx,\vz)\,g(s,\vz,\vy)
 \ \ud z\,\ud s\ ,
\end{equation}
whenever the integrals converge absolutely. The $n$-fold
convolution of $f$ with itself is denoted as $f^{\ast n}=f*\dots *f$.
\end{defn}
The following statement on convolutions will be useful.
\begin{lemma}
\label{Lem:convolest}
Let $\alpha>0$, $\kappa_j>0$ and 
$f_j\in C^\infty\bigl((0,\infty)\times\rz^d\times\rz^d\bigr)$, such
that
\begin{equation}
\label{f_jbound}
 |f_j (t,\vx,\vy)| \leq C_j \, t^{-\frac{d+2}{2}+\kappa_j}\,
 \ue^{-\frac{\alpha}{t}(\vx-\vy)^2}\ ,
\end{equation}
where $C_j>0$. Then 
$f_1*\dots *f_n\in C^\infty\bigl((0,\infty)\times\rz^d\times\rz^d\bigr)$, 
and 
\begin{equation}
\label{f*nest}
 |(f_1*\dots *f_n) (t,\vx,\vy)| \leq D_n\,t^{-\frac{d+2}{2}+\kappa_1+\dots +\kappa_n}\,
 \ue^{-\frac{\alpha}{t}(\vx-\vy)^2}\ ,
\end{equation}
with some $D_n>0$.
\end{lemma}
If the constants $\kappa_j$ are chosen optimally, they can be defined as
a {\it degree} of $f_j$. The unusual definition of the power of $t$ in 
\eqref{f_jbound} therefore leads to an additivity of the degree under 
convolution. From Lemma~\ref{lem:R_Nestimate} one concludes that $R_N$ has 
degree $N+2$.
\begin{proof}
The bounds \eqref{f_jbound} imply that the convolution integrals
defining $f_1*\dots *f_n$ converge absolutely and uniformly, hence
$f_1*\dots *f_n\in C^\infty\bigl((0,\infty)\times\rz^d\times\rz^d\bigr)$.

Moreover, the $n$-fold convolution of the kernels $f_j$, 
\begin{equation}
\begin{split}
 (f_1*\dots *f_n) (t,\vx,\vy) 
  &=  \int_0^t\dots\int_0^{t_{n-2}}\int_{\rz^d}\dots\int_{\rz^d}
         f_1(t-t_1,\vx,\vz_1)f_2(t_1-t_2,\vz_1,\vz_2)\dots  \\
  &\qquad\qquad\dots f_n(t_{n-1},\vz_{n-1},\vy)
        \ \ud z_{n-1}\dots\ud z_1\,\ud t_{n-1}\dots\ud t_1 \ ,
\end{split}
\end{equation}
will be estimated based on the bounds \eqref{f_jbound}. This 
includes a convolution of Gaussians, 
\begin{equation}
\begin{split}
 &\int_{\rz^d}\dots\int_{\rz^d}\ue^{-\alpha\bigl[\frac{(\vx-\vz_1)^2}{t-t_1} 
    +\frac{(\vz_1-\vz_2)^2}{t_1-t_2}+\dots
    +\frac{(\vz_{n-2}-\vz_{n-1})^2}{t_{n-2}-t_{n-1}}
    +\frac{(\vz_{n-1}-\vy)^2}{t_{n-1}}\bigr]} \ \ud z_{n-1}\dots\ud z_1 \\
  &\qquad\qquad\qquad\qquad= \Bigl(\frac{\pi}{\alpha}\Bigr)^{(n-1)d/2}\,
   \left(\frac{(t-t_1)\dots(t_{n-2}-t_{n-1})t_{n-1}}{t}\right)^{d/2}\,
   \ue^{-\frac{\alpha}{t}(\vx-\vy)^2}\ .
\end{split}
\end{equation}
Hence, with $\widetilde{D}_n = (\pi/\alpha)^{(n-1)d/2} \, C_1 C_2
\cdots C_n$ we obtain
\begin{equation}
\begin{split}
 &|(f_1*\dots *f_n)(t,\vx,\vy)| \\
 &\qquad\leq \widetilde{D}_n \,\ue^{-\frac{\alpha}{t}(\vx-\vy)^2}\ \\
  &\qquad\qquad\int_0^t\dots\int_0^{t_{n-2}}
        \frac{(t-t_1)^{\kappa_1-1}\dots(t_{n-2}-t_{n-1})^{\kappa_{n-1}-1}
        t_{n-1}^{\kappa_n-1}}{t^{d/2}}\ \ud t_{n-1}\dots\ud t_1 \\
  &\qquad=  \widetilde{D}_n \,\ue^{-\frac{\alpha}{t}(\vx-\vy)^2}\ 
        t^{-\frac{d+2}{2}+\kappa_1+\dots +\kappa_n}\\ 
  &\qquad\qquad\int_0^1\dots\int_0^{s_{n-2}}(1-s_1)^{\kappa_1-1}\dots
        (s_{n-2}-s_{n-1})^{\kappa_{n-1}-1}s_{n-1}^{\kappa_n-1}\ \ud s_{n-1}
        \dots\ud s_1 \ ,
\end{split}
\end{equation}
which proves the estimate \eqref{f*nest}.
\end{proof}
From this result one obtains the following.
\begin{lemma}
\label{Lem:R_0^nbound}
Let $n\in\nz$, then 
$R_N^{\ast n}\in C^\infty\bigl((0,\infty)\times\rz^d\times\rz^d\bigr)$.
Moreover,
\begin{equation}
\label{R_0^nest}
 |R_N^{\ast n}(t,\vx,\vy)| \leq 
 \Bigl(\frac{4\pi}{1-\varepsilon}\Bigr)^{(n-1)d/2}\,\frac{M_N^n}{(n-1)!}\,
 t^{-\frac{d+2}{2}+n(N+2)}\,\ue^{-\frac{1-\varepsilon}{4t}(\vx-\vy)^2}\ .
\end{equation}
In particular, $R_N^{\ast n}$ has degree $n(N+2)$ and therefore is bounded as 
$t\to 0$ when $n(N+2)\geq\frac{d+2}{2}$.
\end{lemma}
\begin{proof}
With the bound  \eqref{R_0est} this statement follows almost immediately 
from Lemma~\ref{Lem:convolest}. We only need to refine the constant
$D_n$ (in terms of $n$), which includes the contribution
\begin{equation}
\begin{split}
 &\int_0^1\dots\int_0^{s_{n-2}}(1-s_1)^{N+1}\dots(s_{n-2}-s_{n-1})^{N+1}s_{n-1}^{N+1}\ 
       \ud s_{n-1}\dots\ud s_1 \\
  &\qquad\qquad\qquad\qquad\qquad\leq  
       \int_0^1\dots\int_0^{s_{n-2}} \ud s_{n-1}\dots\ud s_1 =\frac{1}{(n-1)!}\ .
\end{split}
\end{equation}
\end{proof}
We also need the following standard result, which is an appropriate version
of Duhamel's principle.
\begin{lemma}
\label{Lem:Duhamel}
Let $f\in C^\infty\bigl((0,\infty)\times\rz^d\times\rz^d\bigr)$, and
$k_N$ be a heat parametrix. Then
\begin{equation}
 \Bigl(\frac{\partial}{\partial t} + H\Bigr)\bigl(k_N *f\bigr) =
 f + R_N *f\ .
\end{equation}
\end{lemma}
One proves this statement by a direct computation. Following a standard
procedure, Duhamel's principle allows us to construct a heat kernel in
terms of a Volterra series.
\begin{theorem}
\label{thm:Volterra}
Let $N\geq d/2 -1$ so that $k_N$ is a heat parametrix. Then the Volterra 
series
\begin{equation}
\label{Voltseries}
 K := k_N + \sum_{n=1}^\infty (-1)^n k_N * R_N^{\ast n}
\end{equation}
converges in $C^\infty\bigl((0,\infty)\times\rz^d\times\rz^d\bigr)$.
It has degree one and defines a heat kernel for the magnetic Schr\"odinger 
operator $H$. 

Moreover, the series \eqref{Voltseries} provides an asymptotic expansion,
as $t\to 0$, for the heat kernel in such a way that the parametrix $k_N$
contributes the first $N+1$ terms in that expansion.
\end{theorem}
\begin{proof}
According to Lemma~\ref{Lem:R_0^nbound}, $R_N^{\ast n}$ is bounded as
$t\to 0$, if $n\geq\frac{d+2}{2N+4}$. Moreover, in analogy to 
Lemma~\ref{lem:R_Nestimate} one finds that
\begin{equation}
 |k_N (t,\vx,\vy)| \leq C_N\,t^{-d/2}\,\ue^{-\frac{1-\varepsilon}{4t}(\vx-\vy)^2}\ ,
\end{equation}
i.e., the degree of $k_N$ is $1$ (independent of $N$). The composition 
Lemma~\ref{Lem:convolest} then implies the bound
\begin{equation}
\label{eq:kNrNnest}
 |k_N\ast R_N^{\ast n}(t,\vx,\vy)| \leq 
 \Bigl(\frac{4\pi}{1-\varepsilon}\Bigr)^{nd/2}\,\frac{C_N M_N^n}{n!}\,
 t^{-\frac{d+2}{2}+n(N+2)+1}\,\ue^{-\frac{1-\varepsilon}{4t}(\vx-\vy)^2}\ .
\end{equation}
Choosing $n\geq\frac{d+2}{2N+4}$ as above, this bound ensures uniform 
convergence (in $\vx,\vy$ and $t\in(0,T)$) of the series. 

The degree of $K$ is determined by $k_N$ to be one; all other terms in the
Volterra series have a higher degree.

Due to the Gaussian factors, any derivative with respect to components of 
$\vx,\vy$ reduces the degree by one, whereas every $t$-derivative reduces 
the degree by two. Hence, for every derivative of $K$ there exists $l\in\nz_0$
such that choosing $n\geq\frac{d+2}{2N+4}+l$ ensures uniform convergence of
the series.
 
As from Lemma~\ref{Lem:Duhamel} we have that 
\begin{equation}
 \Bigl(\frac{\partial}{\partial t}+H\Bigr)k_N * R_N^{\ast n} =
 R_N^{\ast n}+R_N^{\ast(n+1)}\ ,
\end{equation}
the uniform convergence of the Volterra series and its derivatives
implies that after applying $\frac{\partial}{\partial t}+H$ the series
telescopes.  Therefore, $(\frac{\partial}{\partial t}+H)K=0$. The
initial condition as $t\to 0$ was built into the construction of
$k_N$; all other terms give no contribution to this limit.

The series \eqref{Voltseries} is asymptotic for small $t$ in the sense
that $k_N(t,\vx,\vy) $ is a sum of terms with a $t$-dependence (on the 
diagonal $\vx=\vy$) of the form $t^{-\frac{d}{2}+k}$, where $k=0,\dots,N+1$.
The highest power, i.e., the smallest term as $t\to 0$, therefore is
$t^{-\frac{d}{2}+N+1}$. According to \eqref{eq:kNrNnest}, however, the smallest 
power in the remainder, coming from $k_N\ast R_N$, is $t^{-\frac{d}{2}+N+2}$. 
\end{proof}

Informally, Theorem \ref{thm:Volterra}  can be rephrased as
\begin{equation}
 K(t,\vx,\vy) = \frac{1}{(4\pi t)^{d/2}}\,\ue^{-\frac{1}{4t}(\vx-\vy)^2} 
 \sum_{k=0}^{N+1} u_k(\vx,\vy)\,t^k +O(t^{-\frac{d}{2}+N+2}) \ .
\end{equation}
A comparison with \eqref{eq:heatdiagexp} therefore shows that the
heat trace invariants are given by
\begin{equation}
\label{eq:heatinv}
 a_j(\vx) = u_j(\vx,\vx) \ ,
\end{equation}
cf.\ \eqref{eq:loworders} for the lowest orders.

\section{Some applications}
\label{sec:appl}
One can use the knowledge of a complete asymptotic expansion of the
heat kernel for a magnetic Schr\"odinger operator on $\rz^d$ to, e.g.,
compute heat invariants as well as the heat trace asymptotics for
magnetic Schr\"odinger operators on manifolds (with or without
boundary) of the form $\rz^2/\Gamma$, where $\Gamma$ is a discrete
subgroup of the isometry group of $\rz^2$. These include half-planes,
cylinders and tori.  Below we give some examples in dimension
$d=2$. In that case there is only one non-vanishing component
\begin{equation}
\label{eq:2dmagnet}
 B(\vx) = B_{21}(\vx)=-B_{12}(\vx)
\end{equation}
of the field strengths \eqref{eq:Bcompdef}. This is the magnetic field in
$d=2$.
\subsection{Half-plane}
We consider the heat equation for a magnetic Schr\"odinger operator 
\eqref{magSchroe} in the upper half-plane
\begin{equation}
 H_+ := \{ \vx=(x_1,x_2)\in\rz^2 ;\ x_2 >0\}
\end{equation}
with Dirichlet- or Neumann boundary conditions along the boundary.

Following a well-known construction, the heat kernels $K^-(t,\vx,\vy)$
and $K^+(t,\vx,\vy)$ for the upper half-plane with Dirichlet- and
Neumann boundary conditions, respectively, can be constructed from the
corresponding kernel $K(t,\vx,\vy)$ for the entire plane $\rz^2$,
provided the potentials have been extended suitably to the lower
half-plane. Introducing the reflection $R(x_1,x_2)=(x_1,-x_2)$ on
$\rz^2$ across the boundary, we need potentials on $\rz^2$ that
satisfy
\begin{equation}
\label{eq:potentialreflect}
 V(R\vx)=V(\vx)\qquad\text{and}\qquad \vA(R\vx)=R\vA(\vx)\ .
\end{equation}
The latter condition implies that $B(R\vx)=-B(\vx)$. Intuitively, this
ensures that a particle which is reflected at the boundary with reversed
velocity, is subject to a magnetic field with the correct sign.
The conditions \eqref{eq:potentialreflect} can be achieved by  extending 
the potentials from $H_+$ to $\rz^2$ accordingly. In order for the potentials 
to be smooth across the boundary, at $x_2=0$ the odd partial 
$x_2$-derivatives of $V$ and $A_1$ as well as the even partial 
$x_2$-derivatives of $A_2$ must vanish. 

Under these conditions the heat kernel $K(t,\vx,\vy)$ for the extended
problem on $\rz^2$ exists and Theorem \ref{prop:Para} provides heat
parametrices. Then
\begin{equation}
 K^{\mp}(t,\vx,\vy) = K(t,\vx,\vy) \mp K(t,\vx,R\vy)
\end{equation}
are heat kernels for the Dirichlet- or Neumann problem on $H_+$,
respectively. The same construction applies to heat parametrices. 
Their diagonals are
\begin{equation}
\label{eq:DNheatdiag}
 k^\mp_N(t,\vx,\vx) = \frac{1}{4\pi t}\sum_{k=0}^{N+1} 
 \bigl(u_k(\vx,\vx) \mp\ue^{-\frac{x_2^2}{t}}\,u_k(\vx,R\vx)\bigr)t^k \ .
\end{equation}
Hence, following \eqref{eq:heatinv} the heat trace invariants are 
given by
\begin{equation}
\label{heat_invariants_half-plane}
 a_k^{\mp}(\vx) = u_k(\vx,\vx) = a_k(\vx)\ .
\end{equation}
One notices that, due to the second term on the right-hand side of
\eqref{eq:DNheatdiag} the knowledge of heat trace invariants for the
extended problem on $\rz^2$ is not sufficient to construct heat trace
invariants on $H_+$; one rather requires the coefficients of a
parametrix. The resulting heat invariants
\eqref{heat_invariants_half-plane} for the half-plane, however, are
the same as for the entire plane. The reflection in the boundary
produces contributions that are exponentially small and, therefore,
are not visible in the heat trace invariants. This is a manifestation
of Kac' {\it principle of not feeling the boundary} \cite{Kac66}.

It may still be of interest to calculate the first terms of 
\eqref{eq:DNheatdiag} from \eqref{eq:0solution} and \eqref{u_1result}, 
e.g.,
\begin{equation}
 u_0(\vx,\vx) \mp\ue^{-\frac{x_2^2}{t}}\,u_0(\vx,R\vx)
  = 1\mp\exp\Bigl\{-\frac{x_2^2}{t}+
    2\ui x_2 \int_0^1 A_2(x_1,(2s-1)x_2)\ \ud s\Bigr\}\ .
\end{equation}
\subsection{Cylinder}
An infinite cylinder can be realised as the plane $\rz^2$ modulo the
discrete subgroup $\Gamma\cong\gz$ of the euclidean group that is 
generated by the translation $(x_1,x_2)\mapsto (x_1,x_2 +1)$. Functions 
on $\rz^2/\Gamma$ are then functions on $\rz^2$ with 
$f(x_1,x_2)=f(x_1,x_2 +1)$. A scalar potential on the cylinder can be 
extended periodically to $\rz^2$ such that 
$V(x_1,x_2)=V(x_1,x_2 +1)$. The same can be done for the magnetic 
field \eqref{eq:2dmagnet}. The vector potential, however, has to be 
chosen in a particular gauge to allow its components to be periodic, too.
In a Fourier representation of the magnetic field,
\begin{equation}
 B(x_1,x_2) = \sum_{n\in\gz} B_n (x_1)\,\ue^{2\pi\ui nx_2} \ ,
\end{equation}
and of the vector potential
\begin{equation}
 \vA(x_1,x_2) = \sum_{n\in\gz} \vA_n (x_1)\,\ue^{2\pi\ui nx_2} \ ,
\end{equation}
one has to ensure that
\begin{equation}
 B_n (x_1) = 
 \frac{\partial A_{n,2}}{\partial x_1}(x_1) - 2\pi\ui n A_{n,1}(x_1)\ ,
\end{equation}
which is always possible. The magnetic Schr\"odinger operator 
\eqref{magSchroe} with these potentials can then be defined on the 
domain $C_0^\infty(\rz^2/\Gamma)$ on which it is essentially self-adjoint.

The heat kernel for such a magnetic Schr\"odinger operator can be
constructed from the respective kernel for the periodically extended 
problem on $\rz^2$ via
\begin{equation}
 K^\mathrm{cyl}(t,\vx,\vy) = \sum_{n\in\gz} K(t,\vx,(y_1,y_2+n))\ .
\end{equation}
In the same way a heat parametrix results from the respective
parametrix of the periodically extended problem in the plane. Its
diagonal then is
\begin{equation}
\label{eq:cylheatdiag}
 k^\mathrm{cyl}_N(t,\vx,\vx) = \frac{1}{4\pi t}\sum_{k=0}^{N+1} 
 \sum_{n\in\gz} \ue^{-\frac{n^2}{4t}}\, u_k(\vx,(x_1,x_2 +n))\,t^k \ .
\end{equation}
Hence the heat trace invariants for the cylinder are
\begin{equation}
 a_k^\mathrm{cyl}(\vx) = u_k(\vx,\vx) = a_k(\vx)\ ,
\end{equation}
since, again, the corrections coming from summing over the
translations in $\Gamma$ are exponentially small.

The first term in \eqref{eq:cylheatdiag} can be calculated to
give
\begin{equation}
\begin{split}
 \sum_{n\in\gz} \ue^{-\frac{n^2}{4t}}\, u_0(\vx,(x_1,x_2 +n))
  &= \sum_{n\in\gz}\exp\Bigl\{-\frac{n^2}{4t}-\ui n\int_0^1
        A_2(x_1,x_2 +n(1-s))\ \ud s\Bigr\} \\
  &= 1 + \sum_{n\in\gz\backslash\{0\}}\exp\Bigl\{-\frac{n^2}{4t}-
         \ui n\int_0^1 A_2(x_1,x_2 +n(1-s))\ \ud s\Bigr\} \ ,
\end{split}
\end{equation}
thus making the exponentially small correction explicit.
\subsection{Torus}
A torus can be represented as $\rz^2/\Gamma$ in analogy to the
cylinder above, however, with a discrete subgroup $\Gamma\cong\gz^2$ 
of the euclidean group that is generated by the two translations
$(x_1,x_2)\mapsto (x_1,x_2 +1)$ and $(x_1,x_2)\mapsto (x_1 +1,x_2)$.
Functions on $\rz^2/\Gamma$ can be identified with functions on
$\rz^2$ via $f(\vx)=f(\vx+\vn)$ for all $\vn\in\gz^2$. The scalar potential
and the magnetic field can be chosen periodically on $\rz^2$, the
vector potential, however, does not allow such a
choice. A Fourier representation of the magnetic field,
\begin{equation}
 B(\vx) = \sum_{\vn\in\gz} B_{\vn}\,\ue^{2\pi\ui\vn\cdot\vx} \ ,
\end{equation}
and of the vector potential
\begin{equation}
 \vA^\mathrm{per}(\vx) 
 = \sum_{\vn\in\gz} \vA_{\vn} \,\ue^{2\pi\ui\vn\cdot\vx} \ ,
\end{equation}
shows that the constant term $B_0$ (which equals the flux $\Phi$
through the torus) cannot be represented in terms of the periodic
vector potential.  It requires an additional linear
contribution $\vA^0(\vx)$ that generates $B_0$, hence a constant
magnetic field on $\rz^2$. This is related to the flux quantisation on
the torus, i.e., the fact that $B_0=\Phi=2\pi n$, where $n\in\gz$.

The construction of the heat kernel for the magnetic Schr\"odinger
operator, the heat parametrix and the heat invariants is closely analogous
to the case of the cylinder, i.e.,
\begin{equation}
 K^\mathrm{torus}(t,\vx,\vy) = \sum_{\vn\in\gz^2} K(t,\vx,\vy+\vn)\ ,
\end{equation}
leading to the diagonals of parametrices
\begin{equation}
\label{eq:torusheatdiag}
 k_N^\mathrm{torus}(t,\vx,\vx) = \frac{1}{4\pi t}\sum_{k=0}^{N+1}
 \sum_{\vn\in\gz^2}\ue^{-\frac{\vn^2}{4t}}\,u_k(\vx,\vx+\vn)\,t^k \ .
\end{equation}
The heat trace invariants for the torus, therefore, are
\begin{equation}
 a_k^\mathrm{torus}(\vx) = u_k(\vx,\vx) = a_k(\vx)\ ,
\end{equation}
since, as for the cylinder the corrections coming from summing over 
the translations in $\Gamma$ are exponentially small. Explicitly, the
first term in \eqref{eq:torusheatdiag} is
\begin{equation}
\begin{split}
 \sum_{\vn\in\gz^2}\ue^{-\frac{\vn^2}{4t}}\,u_0(\vx,\vx+\vn)
  &=  \sum_{\vn\in\gz^2}\exp\Bigl\{-\frac{\vn^2}{4t}-\ui\int_0^1
         \vn\cdot\vA(\vx+(1-s)\vn)\ \ud s\Bigr\} \\
  &= 1 + \sum_{\vn\in\gz^2\backslash\{0\}}\exp\Bigl\{-\frac{\vn^2}{4t}-
         \ui\int_0^1\vn\cdot\vA(\vx+(1-s)\vn)\ \ud s\Bigr\} \ .
\end{split}
\end{equation}

\vspace*{0.5cm}
\subsection*{Acknowledgements}
We gratefully acknowledge financial support by a Royal Society
International Joint Project grant.

\vspace*{0.5cm}
{\small
\bibliographystyle{beta}
\bibliography{literatur}}

\end{document}